%% file: main.tex
\documentclass[letterpaper, 11pt]{article}
\usepackage[utf8]{inputenc}
\usepackage{xifthen}
\usepackage{graphicx} 
\usepackage{fullpage}
\usepackage[ruled,linesnumbered]{algorithm2e} 
\usepackage{algorithmic}
\usepackage{xspace}

\let\oldnl\nl
\newcommand{\nonl}{\renewcommand{\nl}{\let\nl\oldnl}}

\usepackage{amsmath, amsthm, amssymb}

\usepackage{hyperref}
\hypersetup{
  colorlinks=true,
  linkcolor=blue,
  linktoc=all,
  filecolor=magenta,
  urlcolor=cyan,
  citecolor=blue,
  bookmarksnumbered=true
}

\usepackage[letterpaper, margin=1in]{geometry}

\usepackage{todonotes}
\usepackage{cleveref}

\newtheorem{theorem}{Theorem}
\newtheorem{observation}[theorem]{Observation}

\newtheorem*{claim*}{Claim}

\newtheorem{proposition}[theorem]{Proposition}
\newtheorem{corollary}[theorem]{Corollary}
\theoremstyle{definition}

\newtheorem{remark}[theorem]{Remark}
\newtheorem*{remark*}{Remark}

\newcommand{\markovstep}{{\+A}}
\newcommand{\approxsampler}{{\+A}}

\newcommand{\transitioncompalgo}{{\+B}}
\newcommand{\targetdist}{r}
\newcommand{\outputdist}{p}
\newcommand{\supp}{supp}
\newcommand{\Alg}{\mathrm{Alg}}
\newcommand{\targetcomp}{{\+B}}
\newcommand{\samplercomp}{{\+C}}
\newcommand{\targetcompall}{{\+B}}
\newcommand{\samplercompall}{{\+C}}
\newcommand{\Ber}{Ber}
\newcommand{\gensampler}{{\sc GenPerfectSampler}\xspace}
\newcommand{\gensamplerrs}{{\sc GenPerfectSamplerRS}\xspace}

\newcommand{\poly}{\text{poly}}

\def\oPr{\mathbf{Pr}}
\renewcommand{\Pr}[2][]{ \ifthenelse{\isempty{#1}}
  {\oPr\left[#2\right]}
  {\oPr_{#1}\left[#2\right]} }

\def\oE{\mathbb{E}}
\newcommand{\E}[2][]{ \ifthenelse{\isempty{#1}}
  {\oE\left[#2\right]}
  {\oE_{#1}\left[#2\right]} }

\def\oVar{\mathbf{Var}}
\newcommand{\Var}[2][]{ \ifthenelse{\isempty{#1}}
  {\oVar\left[#2\right]}
  {\oVar_{#1}\left[#2\right]} }

\def\oEnt{\mathbf{Ent}}
\newcommand{\Ent}[2][]{ \ifthenelse{\isempty{#1}}
  {\oEnt\left[#2\right]}
  {\oEnt_{#1}\left[#2\right]} }

  \def\*#1{\boldsymbol{#1}} 
  \def\+#1{\mathcal{#1}} 
  \def\-#1{\mathrm{#1}} 
  \def\^#1{\mathbb{#1}} 
  \def\!#1{\mathfrak{#1}}

\newcommand{\D}[2]{\-D_{(#1)}(#2)}

\renewcommand{\epsilon}{\varepsilon}

\newcommand{\norm}[1]{\left\Vert#1\right\Vert}
\newcommand{\set}[1]{\left\{#1\right\}}
\newcommand{\tuple}[1]{\left(#1\right)} \newcommand{\eps}{\varepsilon}

\newcommand{\abs}[1]{\left\vert#1\right\vert}

\newcommand{\Lovasz}{Lov\'asz\xspace}

\title{Perfect sampling from rapidly mixing Markov chains}
\author{
Andreas G\"{o}bel \thanks{ Hasso Plattner Institute, University of Potsdam, Germany. \textnormal{E-mail:~\url{andreas.goebel@hpi.de}}}
\and
Jingcheng Liu\thanks{State Key Laboratory for Novel Software Technology, Nanjing University, China. \textnormal{E-mail:~\url{liu@nju.edu.cn}}}
\and 
Pasin Manurangsi\thanks{Google Research, Thailand. \textnormal{E-mail:~\url{pasin@google.com}}}
\and
Marcus Pappik \thanks{ Hasso Plattner Institute, University of Potsdam, Germany. \textnormal{E-mail:~\url{marcus.pappik@hpi.de}}}
}
\date{}

\begin{document}

\maketitle
\begin{abstract}
We show that efficient approximate sampling algorithms, combined with a slow exponential time oracle for computing its output distribution, can be combined into constructing efficient perfect samplers, which sample exactly from a target distribution with zero error upon termination. This extends a classical reduction of Jerrum, Valiant and Vazirani, which says that for self-reducible problems, deterministic approximate counting can be used to construct perfect samplers. We provide two surprisingly simple constructions,  and our perfect samplers run in polynomial time both in expectation and with high probability. 

An overwhelming amount of efficient approximate sampling algorithms are based on Markov chains. Informally, we show that any Markov chains with absolute spectral gap $\gamma$ can be converted into a perfect sampler with expected time $O\tuple{\frac{1}{\gamma}\ln\frac{\abs{\Omega}}{\pi_{*}}}$, where $\pi_{*}$ is the minimum probability in the stationary distribution. 
This is also the best possible bound for mixing time to achieve approximate sampling from a spectral gap, and we are able to do perfect sampling in the same time bound in expectation.

We also highlight a number of applications where we either get the first perfect sampler up to the uniqueness regime (roughly speaking, everywhere except where NP-hardness results are known), or the fastest perfect sampler known to date. Remarkably, we are able to get the first perfect sampler for perfect matchings of bipartite graphs based on the celebrated Jerrum-Sinclair-Vigoda algorithm.
\end{abstract}

\begin{center}
\begin{minipage}{0.95\textwidth} 
  \textbf{Funding:} Andreas Göbel was funded by the project PAGES (project No. 467516565) of the German Research Foundation (DFG). 
  Jingcheng Liu is supported by the National Science Foundation of China under Grant No. 62472212.
  Marcus Pappik was funded by the HPI Research School on Data Science and Engineering.
\end{minipage}
\end{center}

\input{intro}

\input{prelim}

\input{generic-sampler}

\input{prelim-markov}

\input{mc-sampler}

\subsection{Simple applications} \label{subsec:simple_applications}

We highlight a number of applications where we either get the first perfect sampler, or get the fastest perfect sampler to date.

\paragraph{Graph $q$-colorings.}
For general graphs of maximum degree $\Delta$, the state-of-the-art result~\cite{vigoda1999improved,chen2019improved} says that if $q\ge (\frac{11}{6} - \eps_0) \Delta $ for some absolute constant $\eps_0 \approx 10^{-5}$, then the flip dynamics has a spectral gap of $\Omega\tuple{1/n}$. The same spectral gap also holds for Glauber dynamics, as their spectral gaps are the same up to constant factors. Combined with our result, we have:
\begin{corollary}\label{cor:col}
    There is a perfect sampler for $q$-colorings on graphs of maximum degree $\Delta$ that runs in expected time $O(n^2)$, provided that $q\ge (\frac{11}{6} - \eps_0) \Delta $ for some absolute constant $\eps_0$.
\end{corollary}

Assuming a girth-assumption, one can further reduce the constant $11/6$, a direction in which there has been a long line of work (see, e.g.,~\cite{jonasson2002uniqueness,frieze2001randomly,molloy2002glauber,dyer2003randomly,dyer2004randomly,dyer2006randomly,gamarnik2015strong,efthymiou2019improved}). Recently, a remarkable breakthrough is made by Chen, Liu, Mani and Moitra~\cite{chen2023strong} that shows a $\Omega(1/n)$ spectral gap provided that $q \ge \Delta + 3$ on graphs of maximum degree $\Delta$ and girth $g(\Delta) = \Omega_{\Delta}(1)$.
Combined with our result, one gets a perfect sampler.
\begin{corollary}\label{cor:col-girth}
    There is a function $g(\Delta)$ such that, for graphs of maximum degree $\Delta$ and girth $g(\Delta) = \Omega_{\Delta}(1)$, there is a perfect sampler that outputs a uniformly random $q$-coloring in expected time $O(n^2)$, provided that $q \ge \Delta + 3$.
\end{corollary}

Prior to our result, the best known perfect sampler either requires $(8/3+o(1))\Delta$ colors~\cite{jain2021perfectly} for general graphs, or is based on deterministic approximation~\cite{liu2022correlation} and requires $\approx 2\Delta$ colors for general graphs and $\approx 1.76\Delta$ colors for triangle-free graphs, where the latter algorithms suffered from a huge polynomial running time of $n^{\exp(\poly(q))}$ (see \cite{liu2019deterministic} for details). Parallel to our work a deterministic approximation algorithm appeared in the literature \cite{chen2024deter}, which, combined with the reduction from \cite{jerrum1986random}, yields a perfect sampler for the same regime as \Cref{cor:col,cor:col-girth}, but with running time in in $n^{f(q,\Delta)}$, where $f\in \Delta^{O(\log\log\Delta)}\log q$.

\paragraph{Spin systems.}
Significant progress has been made in understanding the mixing time of various dynamics for spin systems. 
For ease of exposition, we focus on $2$-spin systems. Given a graph $G=(V,E)$, a configuration $\sigma$ assigns a spin from $\set{0,1}$ to each vertex. Each edge of the graph is associated with a $2\times 2$ interaction matrix $A$, which models the interaction of the spins. Each vertex may also come with a preference on one of the spins (also known as the \emph{external fields}), which we denote by a parameter $\lambda_v$. Formally, we define the weight of a configuration to be 
\[
w(\sigma) = \prod_{e=(u,v)} A_{\sigma(u), \sigma(v)} \prod_{v\in V: \sigma(v)=1} \lambda_v. 
\]
Then, the Gibbs distribution (also known as the Boltzmann distribution) of the spin system is defined over the configurations, with the sampling probability of each configuration being proportional to its weight. Formally,

$\Pr{\sigma} = \frac{1}{Z_G} w(\sigma)$,

where $Z_G$ is the normalizing constant (also known as the partition function) for the spin system.

Without loss of generality, any symmetric interaction $A$ can be written in the form of $A=\begin{pmatrix}
    \beta & 1\\
    1 & \gamma
\end{pmatrix}$. 
Moreover, it is often assumed that $\lambda_v$ takes the same value $\lambda$ for all vertices, meaning that a $2$-spin system is then specified  by three parameters $(\beta,\gamma,\lambda)$.
If $\beta \gamma>1$, the system is known to be in the ferromagnetic regime and tends to favor agreements across edges; and if $\beta \gamma <1$, the system is known as anti-ferromagnetic and favors disagreements. The case of $\beta \gamma = 1$ means there is no interaction across edges. In the case of $\beta=\gamma$ we have the celebrated Ising model, which can be seen as a cut generating polynomial. If $\beta=1$ and $\gamma=0$, this is known as the hardcore model or the independence polynomial (as a generating polynomial of independent sets). If $\lambda = 1$, we say that the model has zero external fields, and non-zero external field if $\lambda\neq 1$. For the Ising model, $\lambda>1$ is symmetric to $\lambda<1$. Without loss of generality, one can assume $\lambda\le 1$. \footnote{This is also true for vertex-dependent external fields, provided that they are consistently in the same direction.}

In the antiferromagnetic case, the uniqueness regime of the model has been well understood and is precisely where efficient sampling and approximate counting algorithms are possible (unless $NP=RP$ as shown by \cite{sly2010computational,sly2012computational,galanis2014improved}). 
For any degree $\Delta \ge 3$, there is a so-called up-to-$\Delta$ unique regime~\cite{li2013correlation} for the parameters $(\beta,\gamma,\lambda)$ where efficient algorithms for approximate counting and sampling are known. 
If the parameters are strictly in the interior of the regime with a slackness $\delta$, this is often called $\delta$-unique~\cite{chen2020rapid,chen2022optimal}. The running time of the approximate counter or sampler often requires $\delta$ to be at least inverse polynomial, or even a constant for some deterministic approximate counting algorithms.
Prior to our work, perfect samplers based on deterministic approximate counting were the only algorithms known to work throughout the uniqueness regime. These algorithms are based on a rejection sampling reduction~\cite[Theorem~3.3]{jerrum1986random}, which also shows that for self-reducible problems, deterministic approximate counting imply existence of perfect samplers. However, a drawback of such constructions is that deterministic approximate counting algorithms often suffer from a huge polynomial runtime, usually of order $n^{\ln \Delta}$ for graphs of maximum degree $\Delta$ (see e.g.,~\cite{weitz2006counting,sinclair2012approximation,li2012approximate,li2013correlation} or~\cite{barvinok2016combinatorics,patel2017deterministic,liu2022correlation}). 
By constructing perfect samplers from Markov chains, this significantly improves the runtime of the resulting perfect sampler.

\begin{theorem}[\cite{chen2020rapid,chen2022optimal}]
Fix any $\delta>0$, degree $\Delta \ge 3$, for any $(\beta,\gamma,\lambda)$ in the up-to-$\Delta$ 
 unique regime of antiferromagnetic $2$-spin model with a $\delta$-unique slackness, and any graph of maximum degree $\Delta$, the absolute spectral gap of the Glauber dynamics satisfy
    $\gamma_* \ge \Omega_{\delta}(1/n)$.
\end{theorem}

\begin{corollary}
    There is a perfect sampler for the antiferromagnetic $2$-spin model in the up-to-$\Delta$ unique regime on graphs of maximum degree $\Delta$, that runs in expected time in $O(n^2)$.
\end{corollary}

We remark that for the hardcore model, there have been numerous attempts in constructing a fast perfect sampler, none of which can go up to the uniqueness threshold (see, e.g.,~\cite{fill2000randomness,guo2019uniform}). Our perfect sampler runs in $O(n^2)$ time instead of linear time, but can work all the way up to the uniqueness threshold with no restrictions on the input graph.

\paragraph{Ferromagnetic Ising model.} We lift existing spectral gap bounds for the subgraph world process~\cite{jerrum1993polynomial} in the presence of non-zero uniform external field $\lambda\neq 1$, and for the random cluster dynamics~\cite{guo2018random} in the case of zero external field. Again, this significantly improves upon existing perfect sampler based on deterministic approximation~\cite{barvinok2016combinatorics,liu2019ising}\footnote{Technically, ferromagnetic Ising model is not self-reducible. However, its equivalent representation, the subgraphs world model is self-reducible~\cite{randall1999sampling}.}.
\begin{theorem}[\cite{jerrum1993polynomial,chen2021spectral}]
Fix any $\lambda\neq 1$, degree $\Delta \ge 3$, for any graph of maximum degree $\Delta$, the absolute spectral gap of the subgraph-world process on the Ising model with non-zero external field satisfy
    $\gamma_* \ge \Omega_{\lambda,\Delta}(1/n)$.
\end{theorem}

\begin{theorem}[\cite{guo2018random}]
    The absolute spectral gap of the random cluster dynamics at $q=2$ satisfy
    $\gamma_* \ge \Omega(\frac{1}{n^4 m^2})$.
\end{theorem}
By noting that $\pi_* \ge \frac{1}{\beta^m \lambda^n}$ for $\lambda>1$, and $\pi_* \ge \frac{1}{\beta^m}$ for zero external fields, and the fact that each step of a random cluster dynamics can be simulated in time $O(m)$, we have that
\begin{corollary}
    There is a perfect sampler for the weighted even subgraphs and for the ferromagnetic Ising model on graphs of maximum degree $\Delta$, and it runs in expected time $O_{\beta,\lambda}(n (m+n))$ for non-zero external field and $O(n^4 m^4 \ln \beta)$ for zero external field.
\end{corollary}

Note that our black-box reduction in the zero external field case achieves similar running times as \cite{feng2022sampling}.

\paragraph{Linear extensions of posets.}
Given a partially ordered relation (poset), a linear extension is a total ordering consistent with the poset. We consider the Markov chains by Bubley and Dyer~\cite{bubley1999faster,wilson2004mixing} who showed $\ell_1$ mixing time of $O(n^3 \ln n)$. 
\begin{corollary}
There is a perfect sampler for linear extensions of posets, with expected time $O(n^4 (\ln n)^2)$.
\end{corollary}

\paragraph{Log-concave sampling and spanning trees.}
Given a distribution $\mu : {[n] \choose k} \to \^R_{\ge 0}$, where $[n] = \{1, \dots, n\}$, it was shown in~\cite{anari2024log,cryan2019modified} that the bases exchange walk associated with $\mu$ has spectral gap $\Omega(1/k)$ if $\mu$ is \emph{strongly log-concave}. In particular, this includes the uniform distribution over the bases of any matroid. Noting that the state space is at most $n^k$, this implies that the bases exchange walk can be converted to a perfect sampler with expected time $O(k^2 \ln n)$.
\begin{corollary}
There is a perfect sampler for homogeneous and strongly log-concave distribution $\mu : {[n] \choose k} \to \^R_{\ge 0}$ with expected time $O(k^2 \ln n)$.
\end{corollary}

In the special case of uniform sampling from spanning trees, our algorithm is not the fastest to-date and runs in time $O(n^2 \ln n)$ on a graph with $n$ vertices, but it has the virtue of using existing Markov chains as a black box and is conceptually simple.
The fastest uniform sampler for spanning trees known to date is the almost-linear time $\tilde{O}(m^{1+\eps})$ algorithm by Schild~\cite{schild2018almost} for an $m$-edge graph, whose algorithm is based on a long line of beautiful research on cleverly shortcutting the random walk by Aldous~\cite{aldous1990random} and Broder~\cite{broder1989generating}.
Even in this special case, our black-box reduction is faster than well-known nontrivial perfect samplers, such as the vanilla Aldous~\cite{aldous1990random} and Broder~\cite{broder1989generating} walk, or the Propp and Wilson's cycle popping algorithm~\cite{wilson1996generating,PW98}.

\section{Perfect matchings in bipartite graphs}
A celebrated result in approximate sampling is the FPAUS for perfect matchings in bipartite graph due to Jerrum, Sinclair and Vigoda (henceforth JSV)~\cite{jerrum2004polynomial}. They designed a sequence of Markov chains on the state space of perfect matchings and near-perfect matchings by cleverly assigning weights to near-perfect matchings, so that the set of perfect matchings accounts for $\Omega(1/n^2)$ of probability in the stationary distribution, and is also uniform when restricted to perfect matchings.
The missing ingredient to apply our perfect sampling reduction is an exponential time oracle to compute the output distribution of this FPAUS. This is a non-trivial step because, every transition matrix along the sequence of Markov chains is constructed by sampling from the previous Markov chains in the sequence, combined with a concentration bound. The rapid mixing of every Markov chains in the sequence also depends on the rapid mixing of the previous ones in the sequence.

Before giving our exponential time oracle, we give a high-level overview of~\cite{jerrum2004polynomial} and explain the technical challenges in designing the exponential time oracle. It is crucial to introduce the correct \emph{transition weights} to near-perfect matchings as in~\cite{jerrum2004polynomial} because, to go from one perfect matching to another, the natural way has to go through near-perfect matchings; but in general, there can be exponentially many more near-perfect matchings than perfect matchings, making it very hard to come back to the set of perfect matchings.
To find these transition weights, they introduce a sequence of instances to ``interpolate'' to the actual graph from a complete bipartite graph. Starting from a complete bipartite graph, it is easy to find the correct transition weights for near-perfect matchings. In each subsequent step, JSV maintain an approximation to the transition weights. To do so, they use the transition weights computed from the previous instance, and run Markov chains to obtain a refined estimate of the transition weights for the current instance. This is made possible because, good approximation to transition weights also guarantees that the spectral gap of the Markov chain will be at least $\Omega(1/n^6)$, which is shown by a clever canonical path construction between the set of near perfect matchings and perfect matchings and routing everything else through these canonical paths. 
A standard Chernoff bound is then applied to $S$ independent runs of Markov chains to conclude that good estimates of the correct weight can be obtained for the current instance with high probability. 
Combined, they showed an FPAUS that runs in ${O}(n^{11} \ln^3(n) \ln^2 (1/\delta))$ to get within $\delta$ in $\ell_1$ distance to the uniform distribution of perfect matchings.

Note that the output distribution of the final Markov chain crucially depends on the distribution of the weights found along the interpolation, which is successively computed by randomized algorithms. A naive attempt to simply enumerate their trajectories would lead to an oracle that is exponential not only in the problem size but also the running time of the FPAUS. One could also try to turn every intermediate Markov chain into a perfect sampler one after another, but the running time also appears to blow up: the mixing time analysis of each chain relies on finding good transition weights from the previous chain, so part of the total variation distance in the end is coming from failing to find good transition weights, which can happen anywhere along the interpolation.

Our perfect sampler will make oracle calls to the FPAUS from~\cite{jerrum2004polynomial} in a black-box fashion. However, to compute its output distribution, we need to look into how their Markov chains work. In the following we start with a summary.

\subsection{A summary of the JSV chains and their analysis}
We will follow the terminology and notation of~\cite{jerrum2004polynomial} as close as possible. We fix a bipartite graph $G=(U \cup V, E)$ throughout the process, and the goal is to sample uniformly at random from the set of perfect matchings of $G$.
To interpolate between graphs, JSV introduce \emph{edge weights} $\lambda_{u,v}$ for every pair of vertices $(u,v) \in U \times V$. We note that edge weights here are defined also for non-edges, as JSV deals with a weighted complete bipartite graph instead of the original graph.
Initially $\lambda_{u,v} = 1$ for every pair of vertices. We need to gradually reduce the weight $\lambda_{u,v}$ for every non-edge $uv\not\in E$ to the very small value of $\frac{1}{n!}$ while keeping $\lambda_e = 1$ for all edges $e \in E$. 

\begin{remark}
\label{rem:target-pi}
    It is sufficient for us to construct a perfect sampler for the final weighted  bipartite graph, following a similar argument as in~\cite{jerrum2004polynomial}:
because there are at most $n!$ invalid matchings, and each can contribute a weight at most $\frac{1}{n!}$. Furthermore, every valid perfect matching remains equally likely and, assuming there is at least one perfect matching, the introduction of invalid matchings only reduces the success probability of returning a perfect matching (and hence increases the expected running time) by a constant factor.
\end{remark}

Let $\+P$ be the set of perfect matchings of the complete bipartite graph, and let $\+N(u,v)$ be the set of near perfect matchings with $u$ and $v$ being the only un-matched vertices.
Matchings in $\+N(u,v)$ are also said to have holes at $u$ and $v$.
Given the edge weights $\set{\lambda_{u,v}}$, each matching $M$ is associated with a weight $\lambda(M) := \prod_{e\in M} \lambda_e$.
To account for insufficient weight of perfect matchings, JSV introduce an additional \emph{transition weight} $w(u,v)$ for any hole pattern $(u,v)$.
Then, the stationary distribution we try to sample from is defined so that $\pi_{\lambda,w}(M) \propto \Lambda(M)$, where:
\begin{align*}
    \Lambda(M):=\begin{cases}
        \lambda(M) w(u,v), \qquad &\hbox{if } M \in \+N(u,v) \hbox{ for some } u,v;\\
        \lambda(M), \qquad\qquad &\hbox{if } M \in \+P.
    \end{cases}
\end{align*}

Given edge weights $\set{\lambda_{u,v}}$ and transition weights $\set{w\tuple{u,v}}$, the Markov chain to sample from $\pi_{\lambda,w}$ consists of proposing three types of moves: removing an edge from a perfect matching uniformly at random; adding an edge to a near perfect matching; exchanging a matching edge with another edge adjacent to the holes in a near perfect matching. These proposals are combined with a Metropolis acceptance rule so that the Markov chain is reversible with respect to $\pi$. 

Ideally, one would like to find transition weight $w(u,v) = w^*(u,v) := \frac{\lambda(\+P)}{\lambda(\+N(u,v))}$ for every hole $u,v$. Because $\sum_{M \in \+N(u,v)} \lambda(M) = \lambda(\+N(u,v))$, one sees that each hole pattern $(u,v)$ (and the set of perfect matchings) is equally likely in the stationary distribution $\pi_{\lambda, w^*}$. In other words, sampling from $\pi$ yields  a perfect matching with probability $\Omega\tuple{\frac{1}{n^2}}$. 

JSV showed that it is also sufficient to find $w$ that approximates $w^*$:
\begin{align}
    \forall u,v, \qquad \frac{w^*(u,v)}{2} \le w(u,v) \le 2w^*(u,v).
    \label{eq:w-approx}
\end{align}
In particular, \cite[Theorem~3.1, Lemma~4.1 and~4.4]{jerrum2004polynomial} showed that, if the transition weights $w$ satisfy condition~\eqref{eq:w-approx}, then the Markov chain to sample from $\pi_{\lambda, w}$ has spectral gap $\Omega(1/n^6)$.

To find these approximate transition weights, JSV setup a sequence of Markov chains, where each will be referred to as a phase. 
We denote the edge weights in phase $i$ by $\lambda^{(i)}_{u,v} $, and transition weights by $w_i$.
The stationary distribution at phase $i$, denoted by $\pi_i$, is defined with respect to $\lambda^{(i)}$ and $w_i$. It is worth noting that the ideal weight at phase $i$, denoted by $w_i^*$, satisfies:
\begin{align}
    w_i^*(u,v) =  w_i(u,v) \frac{\pi_i(\+P)}{\pi_i(\+N(u,v))}.
    \label{eq:w-update}
\end{align}
Given $w_i$ and an FPAUS to sample from $\pi_i$, one can obtain estimate of $w_i^*(u,v)$ by counting how many are perfect matchings and how many are near perfect matchings with holes $u,v$.
We write $\+M_i$ for the Markov chain with weights $\lambda^{(i)}$ and $w_i$, and stationary distribution $\pi_i$. We will also write $\+M_i^T$ for the output distribution of simulating $\+M_i$ for $T$ steps. By \cite[Theorem 3.1]{jerrum2004polynomial}, if $w_i$ satisfies condition~\eqref{eq:w-approx}, then $\+M_i$ is an FPAUS to sample from $\pi_i$. So by choosing $T=O(n^7 \ln n)$, one could get $\D{\infty}{T} \le c^{1/4}$ for $c=6/5$.

Initially, $\lambda^{(1)}_{u,v}=1 $, and computing $w_1$ is easy on a complete bipartite graph. 

In each subsequent phase, $w_i$ is maintained to satisfy condition~\eqref{eq:w-approx}.
We take $S=O(n^2 \ln (n/\hat{\eta}))$ independent samples from $\+M_i^T$ to approximate  $\frac{\pi_i(\+P)}{\pi_i(\+N(u,v))}$, and update $w_{i+1}(u,v)$ using the identity~\eqref{eq:w-update}. By standard concentration bounds, except with probability $\hat{\eta}$, for all all holes $(u,v)$, the $w_{i+1}(u,v)$ simultaneously approximate $w_i^*(u,v)$ within a factor of $6/5$. 
Finally we select a vertex $v$ and reduce $\lambda_{u,v}$ for all non-edges $uv\not\in E$ as follows:
$\lambda^{(i+1)}_{u,v} \gets \lambda^{(i)}_{u,v} \exp(-1/2)$.
Then we move on to the next phase.

To see why condition~\eqref{eq:w-approx} is maintained, notice the following.
\begin{enumerate}
    \item Except with probability $\hat{\eta}$, the $w_{i+1}(u,v)$ simultaneously approximate $w_i^*(u,v)$ within a factor of $6/5$, for all holes $(u,v)$.
    \item  We only update $\lambda_{u,v}$ for non-edges incident to a common fixed vertex $v$ and only updates the weight by a factor of $\exp(1/2)$.
   By further noting that any matching can have at most one edge incident to the vertex $v$, this means that $w_{i+1}^*$ differs from $w_i^*$ by at most a factor of $\exp(1/2)$;
\end{enumerate}
Combined $w_{i+1}$ approximates $w_{i+1}^*$ within a factor of $6\exp(1/2)/5 < 2$, except with probability $\hat{\eta}$.

In total there are $R=O(n^2 \ln n)$ phases, because there are $n$ vertices and each requires $O(n \ln n)$ phases to reduce edge weights to $\frac{1}{n!}$.
By setting $\hat{\eta} = \eta/(n^2 \ln n)$, we get that the final weight approximates the ideal weight within a factor of $2$ except with probability $\eta$. Then we run the final Markov chain $\+M_R$ for $O(T\ln (1/\delta))$ steps again to generate a sample, and output if it is a perfect matching, otherwise report failure and restart the entire procedure.

If a perfect matching is produced at the end, it is distributed within total variation distance $\delta + \eta$ of the stationary distribution $\pi_R$. By setting $\delta = \eta$, we get an FPAUS with running time $O(n^{11} (\ln n)^2 (\ln n + \ln 1/\delta))$ to generate a perfect matching within TV distance $2\delta$. By \Cref{prop:linfty-l1}, setting $\delta = \frac{\eps}{2 n!}$ yields an $(\eps, \infty)$-approximate sampler. Combined, we have \[
T(\approxsampler_\eps) \le O \tuple{n^{11} (\ln n)^2 (n \ln n + \ln 1/\eps)}.
\]

\subsection{Exponential oracle for the distribution of JSV and a perfect sampler}
We compute the distribution on $w_i(u,v)$ for every $i$ and holes $(u,v)$, inductively on $i$.

Recall that, $w_{i+1}$ is obtained from $w_i$ multiplied with a ratio, which is obtained by taking $S$ samples to estimate the ratio of perfect matchings to near perfect matchings with holes $(u,v)$.
We claim that the possible number of values of $w_i(u,v)$ for any hole $(u,v)$ would be at most $(S+1)^{2R}$. To see this, there are $R=O(n^2 \ln n)$ phases in total, each multiplying a number with at most $(S+1)^2$ possible values. We note that $\infty$ can be treated as a single number. 

Combined with the choices of $u,v$, the set of transition weights $\set{w_i(u,v)}$ as a whole has at most $\tuple{(S+1)^{2R}}^{n^2} = \exp(2 R n^2 \ln (S+1))$ possible values.
We make a table of size $\exp( 2 R n^2 \ln (S+1))$ to record the probabilities of getting each set of transition weight at phase $i$, denoted by $\+T_i$. In particular, the probability that at phase $i$, we get a setting of transition weights $\set{w_i(u,v)}$, will be given by  the corresponding entry in $\+T_i(\set{w_i(u,v)})$.
Next, we show how to construct $\+T_{i+1}$ given $\+T_i$.

We first compute the output distribution of $T$ steps of $\+M_i$, denoted by $\+M_i^T$, 
in time $O(T \abs{\Omega}^3)$. In fact, we only need the distribution of $\+M_i^T$ on $n^2+1$ disjoint sets: the set $\+P$, and the sets $\+N(u,v)$ for all holes $(u,v)$. 
We remark that if $\set{w_i(u,v)}$ satisfies condition~\eqref{eq:w-approx}, then the distribution of $\+M_i^T$ on these $n^2+1$  disjoint sets will be close to a uniform distribution: $\Pr[M \sim \+M_i^T]{M \in \+P} \ge \frac{1}{8(n^2+1)}$, and $\Pr[M \sim \+M_i^T]{M \in \+N(u,v)} \ge \frac{1}{8(n^2+1)}$ for any hole $(u,v)$.

We denote $p:=\Pr[M \sim \+M_i^T]{M \in \+P}$ and $n_{u,v}:= \Pr[M \sim \+M_i^T]{M \in \+N(u,v)}$. Notice that we already know these numbers exactly from the distribution of $\+M_i^T$. Then, each sample of $\+M_i^T$ essentially corresponds to drawing a sample from a multinomial distribution with $n^2+1$ outcomes, with their respective probability specified by $p$ and $\set{n_{u,v}}$.

For any fixed pair of transition weights $\set{w_i(u,v)}$ and $\set{w_{i+1}(u,v)}$, we compute $$\Pr[(\+M_i^T)^{\otimes S}]{\set{w_{i+1}(u,v)} \mid \set{w_i(u,v)}}.$$

Suppose that there exist $S$ samples that lead to updating from $w_i$ to $w_{i+1}$. Among these $S$ samples, let $S_P$ be the number of samples that land in the set $\+P$ of perfect matchings, and $S_{u,v}$ be the number of samples that land into the set $\+N(u,v)$. We note that these $n^2+1$ numbers $S_P$ and $S_{u,v}$ are uniquely determined by $\set{w_i(u,v)}$ and $\set{w_{i+1}(u,v)}$, and it only takes $O(n^2)$ time to compute these numbers.
Therefore, from the pair $\set{w_i(u,v)}$ and $\set{w_{i+1}(u,v)}$, we can compute 
$$\Pr[(\+M_i^T)^{\otimes S}]{\set{w_{i+1}(u,v)} \mid \set{w_i(u,v)}} = \frac{S!}{(S_P)! \prod_{(u,v)\in U\times V} (S_{u,v})!} \cdot  p^{(S_P)} \cdot \prod_{(u,v)\in U\times V} n_{u,v}^{S_{u,v}}.
$$

We average these numbers over the distribution $\set{w_i(u,v)}$ as specified by $\+T_i$, we get the table $\+T_{i+1}$, with a running time of $O( T\abs{\Omega}^3 +\abs{\+T_i}^2 \cdot n^2  ) = O\tuple{T\abs{\Omega}^3 + n^2 \exp(4 R n^2 \ln (S+1))} $.

There are a total of $R$ phases, so it takes $O\tuple{RT\abs{\Omega}^3 + R n^2 \exp(4 R n^2 \ln (S+1))}$ to compute the distribution of the final weights. Then for each setting of transition weights $\set{w_R(u,v)}$ we compute the output distribution of $\+M_R$ after simulating for $O(T\ln (1/\delta))$ steps.

In total it takes $O\tuple{\abs{\Omega}^3 T \ln (1/\delta) \cdot \exp(2R n^2 \ln (S+1)) + RT\abs{\Omega}^3 + R n^2 \exp(4 R n^2 \ln (S+1))}$ to compute the output distribution.
Recall that in~\cite{jerrum2004polynomial},  $S=O(n^2 \ln (n/\delta))$, $T=O(n^7 \ln n)$ , $R=O(n^2 \ln n)$ and $\abs{\Omega} \le n! \le \exp\tuple{n \ln n}$. Furthermore, as discussed at the end of last subsection, we set $\delta = \frac{\eps}{2 n!}$ to get an $(\eps, \infty)$-approximate sampler.  Combined, we have
\begin{align*}
    T(\samplercompall_{\approxsampler_\eps}) \le& O\tuple{ n^7 \ln n \ln \frac{2n!}{\eps} \cdot \exp\tuple{10 n^4 \ln n \cdot (\ln n + \ln \ln \frac{2(n+1)!}{\eps}) } } 
\end{align*}

Recalling~\Cref{rem:target-pi}, it suffices to get a perfect sampler for $\pi_R$. So our target distribution $\targetdist$ would be $\pi_R$. Then, the oracle $\targetcompall_{\targetdist}$ for computing $\pi_R$ is a simple one: enumerate all configurations in a bipartite graph in time $O(n!) \le \exp(n \ln n)$ and compute the partition function, so $T(\targetcompall_{\targetdist}) \le \exp(n \ln n)$.

We set $\eps= \exp(-64 n^4 (\ln n)^2)$ in our reduction in \Cref{thm:gen-red}.
Then we have a perfect sampler which runs in expected time
\begin{align*}
    O\left(T(\approxsampler_\eps) + \eps \cdot T(\samplercompall_{\approxsampler_\eps}) + \eps \cdot T(\targetcompall_{\targetdist}) + \eps \cdot |\Omega|\right) \le& O\tuple{n^{15} (\ln n)^4} + O(n^{11} (\ln n)^3) + O(1) + O(1) \\
    \le& O\tuple{n^{15} (\ln n)^4}.
\end{align*}

\section*{Acknowledgments}
Andreas G\"{o}bel and Marcus Pappik express their gratitude to Mark Jerrum and Will Perkins for interesting discussions
and insightful feedback on an earlier version of this paper.

\bibliographystyle{alpha}
\bibliography{refs}

\end{document}

%% file: intro.tex
\section{Introduction}
Perfect sampling is the algorithmic task of sampling exactly from a given distribution. The distributions are often high dimensional: one may think of sampling a spanning tree, or sampling an independent set of a graph uniformly at random. The approximate sampling task refers to sampling approximately uniformly at random, with approximation in terms of total variation distance.

A closely related task to sampling is counting: given a problem in the complexity class NP, how many witnesses (solutions) are there? These problems belong to the class \#P. As an example, given a graph, one wants to count how many independent sets there are. The approximate counting task refers to approximating this number multiplicatively within $1 \pm \eps$.
For many natural classes of problems (formally, \emph{downward self-reducible} problems in NP), the task of approximate sampling and approximate counting are equivalent~\cite{jerrum1986random}. 

The separations between approximate counting and exact counting under standard complexity assumptions have been extensively explored and are fairly well-understood. Many combinatorial counting problems are known to be \#P-hard, yet admit efficient approximate counting algorithms (e.g., counting the number of matchings in a graph, or solutions to DNF formulas). Jerrum, Valiant and Vazirani~\cite{jerrum1986random} initiated a complexity-theoretic study on perfect sampling, in the form of uniform generation of NP witnesses. They showed that: given an NP oracle, there is a BPP algorithm for approximately counting witnesses to any problem in NP (and thus for approximate sampling as well, see also~\cite{sipser1983complexity,stockmeyer1983complexity} and the bisection technique of~\cite{valiant1985np}); given a $\Sigma_2^{\mathrm{p}}$ oracle, there is a ZPP algorithm for perfect sampling. 
The requirement of a $\Sigma_2^{\mathrm{p}}$ oracle was relaxed to an NP oracle as well for perfect sampling by Bellare, Goldreich and Petrank~\cite{bellare2000uniform}, where they also discussed applications of perfect sampling to zero knowledge proof systems.  As such, approximate counting, approximate sampling and perfect sampling are no harder than NP under Turing reductions for \emph{every} problem in \#P.  In summary, there is a complexity separation between exact and approximate counting not only on certain individual problems, but also as a class of problems in terms of their worst case complexity (unless the polynomial hierarchy collapses, as approximate sampling/counting is no harder than $\mathrm{NP}$ under Turing reductions, while exact counting as a class is \#P).  

\begin{align*}
    \mathbf{ApproxCount}\approx\mathbf{ApproxSample} \overset{?}{\subset} \mathbf{PerfectSample} \overset{?}{\subset} \mathrm{NP} \subset \mathbf{ExactCount} = \#P
\end{align*}

The separation between exact and approximate counting leads to the following question: 
\begin{quote}
Is there a complexity separation between the complexity of approximate and perfect sampling? Can there be a distribution with efficient approximate sampler, yet it is (NP-)hard to do perfect sampling?
\end{quote}
To this date, no complexity separations between approximate and perfect sampling is known, whether in any specific problems or in terms of worst-case complexity. In this work we show that, under mild conditions, any approximate sampler could be converted to a perfect sampler almost for free. In particular, this holds true for any Markov chains that can be simulated computationally.

The existing literature on perfect sampling is vast and mainly focused on algorithms for concrete problems. One of the earliest uniform perfect samplers for a \#P-hard problem is given by Jerrum, Valiant and Vazirani~\cite{jerrum1986random} for solutions to a DNF formula, which is \#P-hard to count exactly. They also showed that, for downward self-reducible problems, any efficient deterministic approximate counting algorithm (fully polynomial-time approximation scheme, or FPTAS) can be converted into a perfect sampler using a simple rejection sampler.  Many FPTASs have been obtained based on the decay of correlations~\cite{weitz2006counting,bayati2007simple,sinclair2012approximation,li2012approximate}, the polynomial interpolation method~\cite{barvinok2016combinatorics,patel2017deterministic,liu2019ising}, and the cluster expansion for models in the \emph{low-temperature} regime~\cite{helmuth2020algorithmic}. However, a drawback of such constructions is that deterministic approximate counting algorithms often suffer from a huge polynomial runtime. In fact,  there are numerous instances where the degree of the polynomial is not an absolute constant and require certain instance parameters to be constants (e.g., the maximum degree of a graph).

Many early approximate sampling algorithms are based on Markov chains, which led Aldous and Diaconis~\cite{aldous1986shuffling,aldous1987strong} to introduce the concept of strong stationary times for Markov chains. Roughly speaking, this is a stopping time when the current sample of the Markov chain is distributed exactly as the stationary distribution, and the sample is independent of the stopping time itself. Many constructions of strong stationary times have been studied and a general theory of strong stationary dual construction was also developed by Diaconis and Fill~\cite{diaconis1990strong}. A notable example is the evolving sets introduced by Morris and Peres~\cite{morris2005evolving}, who were able to prove a bound on uniform mixing time of Markov chains based on spectral profile. 

A different line of work begins with the beautiful idea of coupling from the past (CFTP) introduced by Propp and Wilson~\cite{propp1996exact}. Remarkably, where applicable, CFTP can determine on its own how long a Markov chain should run in order to sample exactly from the stationary distribution.
The fact that it does not require the knowledge of any bound on mixing time makes it particularly valuable, as mixing times in general are notoriously hard to estimate (see, e.g., the SZK-hardness~\cite{bhatnagar2011computational}), even if the chain is promised to mix in polynomial time and the starting state is given. Subsequently, there have been numerous ideas based on CFTP for perfect sampling~\cite{fill1997interruptible,kendall2000perfect,wilson2000couple,fill2000randomness}.
However, designing CFTP for Markov chains that are non-monotone can be technically challenging and successful examples are scarce (see also~\cite{haggstrom1999exact,huber1998exact}). 

More recently, there are also active and fruitful endeavors in the context of  “sampling \Lovasz{} local lemma”: how and when can one sample a satisfying assignment uniformly at random in a regime similar to the local lemma. Perfect samplers in this context are primarily based on two approaches: an LP based approach introduced by Moitra~\cite{moitra2019approximate} which yields a deterministic approximate counting algorithm, and the partial rejection sampling scheme by Guo, Jerrum and Liu~\cite{guo2019uniform} that can be seen as a natural extension of the Moser-Tardos algorithm for uniform sampling. Remarkably, a recent breakthrough was obtained by~\cite{Wang2023ASL} for large domain size. 

The distinction between perfect sampling and approximate sampling also plays an important role in differential privacy (DP)~\cite{dwork2006calibrating}. The celebrated exponential mechanism~\cite{mcsherry2007mechanism} provides a generic distribution that, when sampled from perfectly, guarantees \emph{pure} DP; whereas approximately sampling from the exponential mechanism only yields \emph{approximate} DP, which is much less satisfactory from a privacy perspective. 
In this context however, it is possible to perturb the distribution itself to achieve pure DP without actually solving the perfect sampling task. See, for instance, Section 6 of~\cite{bassily2014private} and~\cite{lintractable} for a more recent example.
We also note that, while there are separations on the sample complexity to achieve pure versus approximate DP, they do not imply a computational complexity separation.  

The plethora of perfect sampling results and techniques for specific problems stand in stark contrast to the lack of complexity-theoretic study of perfect sampling and its relation to approximate sampling. While we are not aware of any explicit conjecture, some authors suggested that perfect sampling may not be computationally harder than approximate sampling. Quoting Feng, Guo and Yin~\cite{feng2019perfect}: ``Morally, we believe that efficient perfect sampling algorithms exist whenever efficient approximate samplers exist.'' In this work, we make a huge step towards a rigorous version of that statement, by showing that, under mild assumptions, approximate sampling algorithms can indeed be turned into perfect sampling algorithms at almost no additional running time costs.

\subsection{Our contributions}
Informally, we show that any \emph{fully polynomial almost uniform sampler}\footnote{Note that the word ``uniform'' is there mostly for historical reasons and the term FPAUS is also used for approximate sampling algorithms from non-uniform distributions.} (FPAUS) $\+A$ 
together with an exponential time oracle $\+B$ that computes the output distribution of $\+A$ and an exponential time oracle $\+C$ that computes the target distribution $\pi$, can be combined into a perfect sampler for $\pi$ that runs in expected polynomial time.

To be precise, we recall the definition of an FPAUS~$\+A$ for a distribution $\pi$ from~\cite{jerrum1986random}. On input of size $n$, $\+A$ outputs a sample within total variation distance $\eps$ to $\pi$ in time polynomial in $n$ and $\ln 1/\eps$.\footnote{As a side note, the requirement on running time of an FPAUS to be polynomial in $\ln 1/\eps$ is standard in the literature (see, e.g.,~\cite{jerrum1986random}), and should not be confused with the requirement of an FPRAS that the running time is required to be polynomial in $1/\eps$ and $\ln 1/\delta$ for multiplicative approximation tolerance $\eps$, and failure probability $\delta$.
In fact,  FPRAS and FPAUS, as defined, are equivalent for self-reducible problems~\cite{jerrum1986random}. Intuitively, the failure probability in FPRAS plays the role of total variation distance in FPAUS.} Furthermore, we require that $\+B$ computes the output distribution of $\+A$ for any given error bound $\varepsilon$ in time exponential in $n$  and sub-linear in $1/\varepsilon$.\footnote{In fact, even $\exp(\poly(n))$ and sub-linear in $1/\varepsilon$ would be sufficient.}
Finally, we require that $\+C$ computes $\pi$ in time exponential in $n$. Under these requirements we show that we can combine the oracles $\+B$ and $\+C$ together with the FPAUS $\+A$ for a sufficiently small choice of $\varepsilon$ to obtain a perfect sampler with expected polynomial running time.


In particular, we show that any Markov chain with a uniform mixing time $T$, can be converted into a perfect sampler in time $O(T \ln \abs{\Omega})$ where $\Omega$ is the state space of the Markov chain. For general Markov chains, uniform mixing time could be obtained by a conductance-based approach, or an $\ell_1$ mixing bound but with a loss of $\ln \pi_*^{-1}$ for the smallest probability in the stationary distribution $\pi$. For reversible Markov chains, uniform mixing time is captured by the spectral gap of the Markov chain, also known as Poincar\'e’s inequality. As a result, we are able to obtain significantly faster perfect samplers, or the first sampler that works all the way up to where approximate counting is possible (uniqueness threshold), all based on a black-box lifting of existing mixing time results on Markov chains.

For ease of exposition, we make two assumptions on the Markov chains in consideration: every step of the Markov chain can be simulated in constant time; 
and the entire transition matrix can be computed in time $O(\abs{\Omega}^4)$, where $\Omega$ is the state space. We note that these assumptions are without loss of generality: for the first one, our expected time bound can be interpreted as how many oracle calls are needed for the Markov chain;  and the second one is arbitrary, by adjusting parameters appropriately one can get identical result for any polynomial time in $\abs{\Omega}$.

\begin{theorem}[Main theorem]
\label{thm:main}
    Given a Markov chain $P$ on the state space $\Omega$ with a stationary distribution $\pi$. 
    If $P$ mixes in time $T$ to within $1/4$ in total variation distance of $\pi$, then there is a perfect sampler for $\pi$ with expected running time $O\tuple{T \ln \frac{\abs{\Omega}}{\pi_*}}$, where $\pi_* = \min_{x\in \supp(\pi)} \pi(x)$. 

    Furthermore, if $P$ is reversible and has absolute spectral gap $\gamma_*$, then there is a perfect sampler for $\pi$ with expected running time $O\tuple{\frac{1}{\gamma_*} \ln \frac{\abs{\Omega}}{\pi_*}}$.
\end{theorem}

We remark that our running time bound is fairly tight assuming one only gets oracle access to the Markov chains and the stationary distribution. 

In general, there appears to be a gap in how many steps are needed in Markov chains for approximate sampling and perfect sampling: consider a $1/2$-lazy random walk on a complete graph: it has a constant spectral gap and it only takes constant time to get within $1/4$ away from the total variation distance. However, reducing the probability of the starting vertex to $O(1/n)$ takes $\Omega(\ln n)$ steps. Assuming only a black box access to the Markov chain, one needs a factor $\alpha$ approximation in order to apply a rejection sampling with success probability $1/\alpha$. So the extra $O(\ln \abs{\Omega})$ factor in running time is needed in the worst case and we show that this is also sufficient.

One may argue that natural Markov chains are not complete graphs but bounded degree expanders, hence there may be more clever tricks for reducing the probability of the starting vertex-state. However, the above example can also be modified: consider two expanders on $n$ and $2^n$ vertices joined by an edge and, say, we start the random walk from the $n$-vertex expander. Then, it takes $O(n)$ steps to escape to the larger expander, but afterwards it takes only $O(n)$ steps to mix in total variation distance. However, if one wants the staying probability within the $n$-vertex expander to be exponentially small, it will take $\Omega(n^2)$ steps. In this case we also have a gap of $O(\ln \abs{\Omega})$. The essence of this gap will become clear when we discuss uniform mixing times.

\subsubsection*{Applications}
We apply our reduction to a number of existing Markov chains and beyond. Most notably,  the celebrated algorithm of Jerrum, Sinclair and Vigoda~\cite{jerrum2004polynomial} for sampling perfect matchings in a bipartite graph is not itself a Markov chain, but is based on a sequence of Markov chains, each constructed with the help of the previous Markov chains in the sequence. We show that there is an exponential time oracle for computing the output distribution of the JSV algorithm, which allows us to get a perfect sampler by invoking the JSV algorithm as black box.
\begin{theorem}
    There is a perfect sampler for perfect matchings in a bipartite graph, with expected running time in $O\tuple{n^{15} (\ln n)^4}$.
\end{theorem}
We emphasize that this is the first perfect sampling algorithm for uniform random perfect matchings on general bipartite graphs with expected polynomial running time.
In particular, no deterministic approximate counting algorithm for the number of perfect matchings is known in this setting, meaning that the reduction from perfect sampling to deterministic approximate counting given by Jerrum, Valiant and Vazirani~\cite{jerrum1986random} is not applicable. 
To the best of our knowledge, perfect sampling algorithms for bipartite perfect matchings only exist in the setting of dense bipartite graphs \cite{huber2006exact} and our result is the first to avoid this restriction.

Besides sampling perfect matchings in bipartite graphs, we showcase our reduction on various other applications.
While we defer a detailed discussion of each application to \Cref{subsec:simple_applications}, we briefly summarize our results here:
\begin{itemize}
     \item We show that, on graphs of maximum degree $\Delta$, there exists perfect sampling algorithm for uniform random \textbf{$q$-colorings} for every $q \ge \left(\frac{11}{6} - \varepsilon_0\right) \Delta$ for some $\varepsilon_0 \approx 10^{-5}$  with expected running time in $O(n^2)$, essentially matching the $O(n^2)$ mixing time known to hold in the same regime for approximate sampling via Glauber dynamics~\cite{chen2019improved}.
    The best known bound on the minimum number of colors required for perfect sampling is obtained via reduction to deterministic approximate counting  \cite{jerrum1986random}. Until recently, the best deterministic approximation was applicable for $q \ge (2 - \varepsilon_1) \Delta$ colors for some small $\varepsilon_1 > 0$ and runs in $n^{\exp(\poly(q))}$~\cite{liu2022correlation,bencs2024deterministic}. 
    Parallel to our work, Chen, Feng, Guo, Zhang, and Zou~\cite{chen2024deter} showed a deterministic approximation algorithm for $q \ge \left(\frac{11}{6} - \varepsilon_0\right) \Delta$ colors with a running time of $n^{f(q,\Delta)}$, where $f\in \Delta^{O(\log\log\Delta)}\log q$.
    While the latter result gives the first perfect sampling algorithm that matches the number of colors required by approximate samplers, the running time is still substantially worse.
    On the other hand, a long line of research was devoted to obtaining perfect sampling algorithms with running times of the form $n^c$ for some $c$ independent of $\Delta$ and $q$  \cite{huber1998exact,feng2022perfect,bhandari2020improved,jain2021perfectly}, where the best known result is a sampler for $q \ge (8/3+o(1))\Delta$ colors with linear running time dependency on $n$ \cite{jain2021perfectly}.
    Our result is the first perfect sampler that fulfills both criteria simultaneously: the best known bound on $q$ and the fastest expected running time. 
    Moreover, our reduction immediately gives improved bounds for more restrictive graph classes, such as providing a perfect sampling algorithm for $q \ge \Delta + 3$ colors on graphs of sufficiently large girth.
    \item We show that for every \textbf{anitferromagnetic $2$-spin model} in the up-to-$\Delta$ uniqueness regime, there exists a perfect sampling algorithm with expected running time in $O(n^2)$ on graphs of size $n$ and maximum degree $\Delta$. Previous perfect sampling algorithms were only applicable to more restrictive parameter regimes~\cite{fill2000randomness,guo2019uniform}, required sub-exponential growth of the graph~\cite{anand2022perfect,feng2022perfect} or had a significantly worse running time of $n^{O(\log\Delta)}$ due to relying on deterministic approximation algorithms (see e.g.,~\cite{weitz2006counting,sinclair2012approximation,li2012approximate,li2013correlation} or~\cite{barvinok2016combinatorics,patel2017deterministic,liu2022correlation}).
    \item We obtain a perfect sampling algorithm for the \textbf{weighted even subgraphs} and for the \textbf{ferromagnetic Ising model} at inverse temperature $\beta$ and external field $\lambda$ on bounded-degree graphs with expected running time $O_{\beta, \lambda}(n(m+n))$ for $\lambda \neq 1$ and with expected running time $O(n^4 m^4 \ln \beta)$ for $\lambda = 1$ (here, $m$ denotes the number of edges of the graph). Hence our black-box reduction in the case of non-zero external field gives a similar running time as the recent monotone CFTP (Coupling From The Past) approach tailored for the random-cluster model by Feng, Guo and Wang~\cite{feng2022sampling}.
    \item We prove that there is a perfect sampling algorithm for \textbf{linear extensions of posets} with expected running time in $O(n^4 (\ln n)^2)$.
    \item We provide a prefect sampling algorithm for \textbf{homogeneous and strongly log-concave distributions} on $\binom{[n]}{k}$ with expected running time in $O(k^2 \ln n)$.
    This result includes a perfect sampling algorithm for the uniform distribution over the bases of a rank $k$ matroid on $n$ elements.
    In the special case of uniform sampling of spanning trees our black-box reduction yields an expected running time of $O(n^2 \ln n)$, which, while being sub-optimal, is still faster than various well-known nontrivial algorithms~\cite{aldous1990random,broder1989generating,wilson1996generating,PW98}.   
\end{itemize}

\subsection{Research directions}

Our reductions follow the recent literature which defines efficiency for perfect sampling algorithms based on expected running time \cite{fill2000randomness,kendall2000perfect,wilson2000couple,huber2006exact,guo2019uniform,jain2021perfectly,feng2022perfect,anand2022perfect}. Some earlier papers, such as \cite{jerrum1986random}, suggest stronger notions of efficiency. An interesting open question is if we can turn an FPAUS into an efficient perfect sampling algorithm under this stronger running time requirement.\footnote{Note that, at first glance, it is not even clear if the existence of a perfect sampling algorithm with polynomial expected running time is a stronger property than the existence of an FPAUS. However, for downward self-reducible problems, this is indeed the case. By Markov's inequality, we can use the perfect sampling algorithm to obtain an $\varepsilon$-approximate sampler with running time $\poly(n, 1/\varepsilon)$. This suffices to obtain an FPRAS for the associated counting problem, which in turn gives an FPAUS via the reduction in \cite{jerrum1986random}.
In fact, every $\varepsilon$-approximate sampler for a self-reducible problem with running time $\poly(n, 1/\varepsilon)$ can be boosted to a running time of $\poly(n)$ and $\log(1/\varepsilon)$.}

In this article our results are focused on distributions of discrete state spaces.
However, Markov chains have been used for approximate sampling from continuous distributions. We believe that, under a suitable model of computation, our results could be extended to continuous domains. Characteristic applications to aim for include sampling from convex bodies in $\mathbb{R}^d$ \cite{DyerFK91,lovasz1993random} or exact simulation of Gibbs point processes \cite{guo2021perfect,anand2023perfect}.
While the theory of $\ell_2$-mixing that we use to obtain more efficient algorithms from reservable Markov chains does extend beyond discrete domains, new ideas are needed to cope with the size of the state space (e.g., through a suitable warm start) and to design the required exponential time oracles (or avoid relying on them).

Finally, we would like to point out that there are  desirable properties of perfect samplers obtained through other techniques, that may not come for free (if possible at all) only assuming a black-box access to an approximate sampler. For example, the perfect sampling algorithm for spin systems on graphs with sub-exponential growth by Anand and Jerrum \cite{anand2022perfect} allows for sampling from finite-volume projections of infinite-volume Gibbs measures under strong spatial mixing. 
It would be interesting to see if their technique could be adapted to turn an approximate sampler for finite projections into a perfect sampler. In particular, it is unclear what is the least assumption under which one can find an ``exponential time oracle'' to approximately compute the marginals of an infinite spin system with a running time dependency on $\eps$ that scales like $o(1/\eps)$, and how to obtain an exact oracle for finite projections of infinite-volume Gibbs measures. 

\subsection{Organization} We present two surprisingly simple constructions of perfect samplers. In Section 2 we introduce two notions of approximations that requires different reductions. In Section 3, we give the two reductions for the two different notions of approximation, by showing how FPAUSs can achieve both notions with slightly different choices of parameters. Our first reduction is based on a simple Metropolis filter, and our second reduction is based on rejection sampling. In Section 4, we focus on Markov chains. In this context, both approximation notions needed for the reduction is captured by the uniform mixing of Markov chains. We first review basic properties in mixing times of Markov chains, then show how standard $\ell_1$ and $\ell_2$ mixing can be used to deduce uniform mixing (also known as $\ell_\infty$ mixing).
In Section 5, we apply our perfect sampling reductions in the context of Markov chains.
In Section 6, we show a perfect sampler for perfect matchings of bipartite graphs.


%% file: prelim.tex
\section{Preliminaries}

Before we go into specific results for Markov chains, we start by describing the generic reduction that works for generic samplers that are possibly not based on Markov chains.

To do this, we will need to establish additional notation.
For two distributions $p, r$, we write $D_{\max}(p, r)$ to denote $\max_{x \in \supp(p)} \frac{p(x)}{r(x)}$ and $D_{(\infty)}(p, r)$ to denote $\max_{x \in \supp(p)} \left|\frac{p(x)}{r(x)} - 1\right|$. 

There are a couple of observations that will be useful to keep in mind. First is a relationship between the two distances, as stated below. 
\begin{observation} \label{remark:inf-vs-max}
If $D_{(\infty)}(p, r) \leq \eps$, then $D_{\max}(p, r) \leq 1 + \eps$.
\end{observation}
\noindent We note that the converse is not true because $D_{\max}(p, r) \leq \eps$ only gives an upper bound of $1 + \eps$ on $\frac{p(x)}{r(x)}$ but, for $D_{(\infty)}(p, r)$ to be small, there must also be a lower bound of $1 - \eps$ on $\frac{p(x)}{r(x)}$ too.

The second observation is that, when $D_{\max}(p, r) \leq \eps$, we can write $r$ as a mixture between $p$ and another distribution $h$, such that the weight of $h$ is at most $\eps$.

\begin{observation} \label{remark:decomp}
If $D_{\max}(p, r) \leq \eps$, then there exists a distribution\footnote{The distribution $h$ is described explicitly in \Cref{eq:generic-perfect-sampler}.} $h$ such that $r = \frac{1}{1 + \eps} \cdot p + \frac{\eps}{1 + \eps} h$
\end{observation}

For a target distribution $\targetdist$, we say that a sampler $\approxsampler_\eps$ is an $(\eps, \max)$-approximate sampler (resp. $(\eps, \infty)$-approximate sampler) iff its output distribution $\outputdist_{\approxsampler_\eps}$ satisfies $D_{\max}(\outputdist_{\approxsampler_\eps}, r) \leq 1 + \eps$ (resp. $D_{(\infty)}(\outputdist_{\approxsampler_\eps}, \targetdist) \leq \eps$). Note that \Cref{remark:inf-vs-max} implies the following.
\begin{observation} \label{remark:inf-vs-max-sampler}
If $\approxsampler_\eps$ is an $(\eps, \infty)$-approximate sampler, then it is also an $(\eps, \max)$-approximate sampler.
\end{observation}

For any algorithm $\Alg$, we write $T(\Alg)$ to denote its expected running time.

%% file: generic-sampler.tex
\section{Perfect Sampling from Generic Sampler}

We are now ready to present our two generic reductions. The starting point for both is to use \Cref{remark:decomp} to decompose a distribution $r$ into a convex combination of a computationally cheap part $p$ and an expensive part $h$.
We then flip a coin with success probability corresponding to the contribution of each of the two distributions to decide with of the two we need to sample from.
Provided the probability that we need to sample from the expensive part is sufficiently small, we still obtain polynomial expected running time.
We remark that similar distribution decompositions have been used as building blocks for application-specific samplers, such as the recently introduced ``lazy depth-first'' sampling approach by Anand and Jerrum~\cite{anand2022perfect} on lattice spin systems.

Our two reductions apply the decomposition in two different ways.
The first reduction decomposes the target distribution directly, using an $(\varepsilon, \text{max})$-approximate sampler for the computational cheap part $p$.
The second reduction is a rejection sampling scheme that uses an $(\varepsilon, \infty)$-approximate sampling algorithm for the proposal distribution and applies the decomposition to the acceptance coin, resulting in better running time for large state spaces at the cost of assuming oracle access to a stronger approximate sampler.
In both reductions, exponential time oracles for the target distribution and the output of the approximate sampler are used for the expensive case.




\subsection{First Reduction}

We present our first reduction in \Cref{eq:generic-perfect-sampler}. This reduction can convert any $(\eps, \max)$-approximate sampler $\approxsampler_\eps$ into a perfect sampler. The idea is very simple: using \Cref{remark:decomp}, we can output a sample from $\approxsampler_\eps$ directly with probability $\frac{1}{1 + \eps} \geq 1 - \eps$. With the remaining probability $\frac{\eps}{1 + \eps}$, we compute $h$ and output a sample from $h$. The crux here is that, while computing $h$ is potentially computationally expensive, this only happens with probability $\eps$. Thus, by setting $\eps$ to be sufficiently small, we can make the running time polynomial in expectation.

The full algorithm is presented in \Cref{eq:generic-perfect-sampler}. Note that the distribution $h$ defined on Line~\ref{line:remainder-dist} of \Cref{eq:generic-perfect-sampler} is valid because $D_{\max}(\outputdist_{\approxsampler_\eps}, \targetdist) \leq 1 + \eps$.

\begin{algorithm}[ht]
\caption{\gensampler}
\label{eq:generic-perfect-sampler}
\nonl \textbf{Requires: } 
\begin{itemize}
\item An $(\eps, \max)$-approximate sampler algorithm $\approxsampler_\eps$ for target distribution $\targetdist$, and,
\item An algorithm $\targetcompall_{\targetdist}$ that outputs $\targetdist(x)$ for all $x \in \Omega$ and,
\item An algorithm $\samplercompall_{\approxsampler_\eps}$ that outputs $\outputdist_{\approxsampler_\eps}(x)$ for all $x \in \Omega$.
\end{itemize}
Sample $y$ from $\Ber\left(\frac{1}{1 + \eps}\right)$ \\
\eIf{y = 1}{
Sample $x$ via $\approxsampler_\eps$ \\
\Return $x$
}{
Compute $\targetdist$ from $\targetcompall_{\targetdist}$ \\
Compute $\outputdist_{\approxsampler_\eps}$ via $\samplercompall_{\approxsampler_\eps}$ \\
\For{$z \in \Omega$}{
Let $h(z) \gets \frac{1 + \eps}{\eps} \left(\targetdist(z) - \frac{1}{1 + \eps} \cdot \outputdist_{\approxsampler_\eps}(z)\right)$ \label{line:remainder-dist}
}
Sample $x \sim h$ \\
\Return $x$
}
\end{algorithm}

The properties of the above sampler are stated and proved below.

\begin{theorem} \label{thm:gen-red}
\gensampler is a perfect sampler for $\targetdist$ with expected running time 
\[
O\left(T(\approxsampler_\eps) + \eps \cdot T(\samplercompall_{\approxsampler_\eps}) + \eps \cdot T(\targetcompall_{\targetdist}) + \eps \cdot |\Omega|\right).
\]

\end{theorem}

\begin{proof}
To see that this is a perfect sampler, note that the output probability is the mixture distribution $\frac{1}{1 + \eps} \cdot \outputdist_{\approxsampler_\eps} + \frac{\eps}{1 + \eps} \cdot h$, which is exactly $\targetdist$ as desired.

The expected running time of the algorithm is
\begin{align*}
&O\left(\frac{1}{1 + \eps} \cdot T(\approxsampler_\eps) + \frac{\eps}{1 + \eps} \left(T(\targetcompall_{\targetdist}) + T(\samplercompall_{\approxsampler_\eps}) + |\Omega|\right)\right) \\
&\leq O\left(T(\approxsampler_\eps) + \eps \cdot T(\samplercompall_{\approxsampler_\eps}) + \eps \cdot T(\targetcompall_{\targetdist}) + \eps \cdot |\Omega|\right). \qedhere
\end{align*}
\end{proof}
\subsection{Second Reduction}

Our second reduction starts from an $(\eps, \infty)$-approximate sampler $\approxsampler_\eps$. Note that, by \Cref{remark:inf-vs-max-sampler}, this is a stronger assumption compared to that of the first reduction. However, the benefit of this reduction is that, unlike \Cref{eq:generic-perfect-sampler} which requires the computation of $\targetdist(x), \outputdist_{\approxsampler_\eps}(x)$ for \emph{all} $x \in \Omega$, our second reduction only needs to compute $\targetdist(x), \outputdist_{\approxsampler_\eps}(x)$ for a \emph{single} sample $x \in \Omega$ produced by $\approxsampler_\eps$ (in each iteration). The latter can be more efficient when the state space is large.

The rough high-level idea is that we use rejection sampling based on samples produced by $\approxsampler_\eps$. That is, we will accept a sample $x$ (produced from $\approxsampler_\eps$) with probability $\frac{r(x)}{(1 + \eps) \cdot \outputdist_{\approxsampler_\eps}(x)}$. Now, since $\approxsampler_\eps$ is an $(\eps, \infty)$-approximate sampler, we know that this acceptance probability is at least $\frac{1 - \eps}{1 + \eps}$. Even though we do not know the exact acceptance probability, we can still accept with this $\frac{1 - \eps}{1 + \eps}$ probability first. Then, with the remaining $\frac{2\eps}{1 + \eps}$ probability, we actually compute the exact acceptance probability and accept accordingly. Although the latter computation can be expensive, it is only required with probability $O(\eps)$ and, thus, when $\eps$ is sufficiently small, the expected running time is small. 

Our reduction is given below in \Cref{eq:generic-perfect-sampler-rs}. It should be noted that the success probability for the Bernoulli distribution on Line~\ref{line:z-bern-prob} is valid because $D_{(\infty)}(\outputdist_{\approxsampler_\eps}, \targetdist) \leq \eps$ implies that $\frac{\targetdist(x)}{\outputdist_{\approxsampler_\eps}(x)} \geq \frac{1}{1 + \eps} \geq 1 - \eps$.

\begin{algorithm}[ht]
\caption{\gensamplerrs}
\label{eq:generic-perfect-sampler-rs}
\nonl \textbf{Requires: } 
\begin{itemize}
\item An $(\eps, \infty)$-approximate sampler algorithm $\approxsampler_\eps$ for target distribution $\targetdist$, and,
\item An algorithm $\targetcomp_{\targetdist}$ that, for any given $x \in \Omega$, computes $\targetdist(x)$ and,
\item An algorithm $\samplercomp_{\approxsampler_\eps}$ that, for any given $x \in \Omega$, computes $\outputdist_{\approxsampler_\eps}(x)$.
\end{itemize}
\While{no sample has been returned}{
Sample $x$ via $\approxsampler_\eps$ \\
Sample $y$ from $\Ber\left(\frac{1 - \eps}{1 + \eps}\right)$ \\
\eIf{y = 1}{
\Return $x$
}{
Compute $\targetdist(x)$ from $\targetcomp_{\targetdist}$ \\
Compute $\outputdist_{\approxsampler_\eps}(x)$ via $\samplercomp_{\approxsampler_\eps}$ \\
Sample $z \sim \Ber\left(\frac{1}{2\eps} \cdot \left(\frac{\targetdist(x)}{\outputdist_{\approxsampler_\eps}(x)} - (1 - \eps)\right)\right)$ \label{line:z-bern-prob} \\ 
\If{$z = 1$}{
\Return $x$
}
}
}
\end{algorithm}

We formalize the properties of the algorithm in \Cref{thm:gen-red-rs}.

\begin{theorem} \label{thm:gen-red-rs}
For any\footnote{$1/2$ here is arbitrary, it could be changed to any constant strictly less than 1.} $\eps \leq 1/2$, \gensamplerrs is a perfect sampler for $\targetdist$ with expected running time $O\left(T(\approxsampler_\eps) + \eps \cdot T(\samplercomp_{\approxsampler_\eps}) + \eps \cdot T(\targetcomp_{\targetdist})\right)$
\end{theorem}

\begin{proof}
We first argue that \gensamplerrs is a perfect sampler for $\targetdist$. To see this, let us consider each iteration of the while-loop. The probability that any $x \in \Omega$ is returned in exactly
\begin{align*}
\outputdist_{\approxsampler_\eps}(x) \cdot \left(\frac{1 - \eps}{1 + \eps} + \frac{2\eps}{1 + \eps} \cdot \left(\frac{1}{2\eps} \cdot \left(\frac{\targetdist(x)}{\outputdist_{\approxsampler_\eps}(x)} - (1 - \eps)\right)\right)\right) = \frac{r(x)}{1 + \eps}.
\end{align*}
Thus, this is the perfect sampler for $r$ as desired.

As for the running time, since each iteration results in terminating with probability at least $\frac{1 - \eps}{1 + \eps} \geq \frac{1}{3}$, the expected running time is at most $O(1)$ times that of the running time of each iteration. Since $\samplercomp_{\approxsampler_\eps}$ and $\targetcomp_{\targetdist}$ are invoked in each iteration with probability $1 - \frac{1 - \eps}{1 + \eps} \leq 2\eps$, the latter is at most $O\left(T(\approxsampler_\eps) + \eps \cdot T(\samplercomp_{\approxsampler_\eps}) + \eps \cdot T(\targetcomp_{\targetdist})\right)$. This yields the desired running time claim.
\end{proof}

%% file: prelim-markov.tex
\section{Uniform mixing times of finite Markov chains} 
Before applying our generic perfect sampler to Markov chains, we first set up basic notations and recall some standard facts on mixing times.
For simplicity, we will focus on finite state spaces.

\subsection{Definitions and preliminaries}
A \emph{finite Markov chain} is a stochastic process $(X_t)_{t \in \mathbb{N}_0}$ on a \emph{finite state space} $\Omega$, which evolution is dictated by a \emph{transition matrix} $P\in [0,1]^{\Omega \times \Omega}$.
Given a starting state $X_0 \in \Omega$, we sample $X_{t+1} \sim P(X_t, \cdot)$. 
A standard example is a random walk on a finite graph.
Denoting the probability $\Pr{X_{t} =y \mid X_0=x}$ by $p^t(x,y)$, it is easy to verify that $p^t(x,\cdot) = \boldsymbol{1}_x^\top P^t$, where $\boldsymbol{1}_x$ is the vector that is $1$ at $x$, and $0$ otherwise.


\paragraph{Stationary distributions.}
For finite $\Omega$, a distribution $\pi$ is \emph{stationary} if $\pi P = \pi$. If a distribution $\mu$ satisfies the \emph{detailed balance equation} 
\[
\mu(x)P(x,y) = \mu(y) P(y,x), \forall x,y\in \Omega,
\]
then we say $P$ is \emph{reversible} with respect to $\mu$. In fact,  a distribution $\mu$ satisfying the detailed balance equation must also be stationary. 


\paragraph{$\ell_p$ distances.}
Given a distribution $\pi$ and $p \in [0, \infty]$, we define the $\ell_p(\pi)$ norm of a function $f : \Omega \to \^R$ by
\begin{align*}
\norm{f}_p :=& \tuple{\sum_{x\in \Omega} \pi(x) \abs{f(x)}^p}^{1/p} = \E[x\sim \pi]{\abs{f(x)}^p}^{1/p}, \quad\hbox{ for finite }p,\\
\norm{f}_\infty :=& \max_{x\in \Omega} |f(x)|.
\end{align*}
For a transition matrix $P$ with stationary distribution $\pi$, we define 
$q_t(x,y):=\frac{p^t(x,y)}{\pi(y)}$.

Note that the expectation of $q_t(x,y)$ is $1$, as 
$\E[y\sim \pi]{q_t(x,y)} = \sum_{y\in \Omega} q_t(x,y) \pi(y) = 1$.
Further, for reversible Markov chains $q_t$ is symmetric.

The $\ell_p$ distance $\D{p}{t}$ is then defined as
$\D{p}{t}:= \max_{x \in \Omega} \norm{q_t(x,\cdot) - 1}_p$.

\begin{remark}
A couple of remarks are in order. 
    \begin{itemize}
        \item For $p=1$, 
$\D{1}{t}$ captures the standard total variation (TV) distance up to a constant $2$ as 
\[
\D{1}{t} = \max_{x \in \Omega} \sum_{y\in \Omega} \abs{p^t(x, y) - \pi(y)}  .
\]
In particular, we may define the $\varepsilon$-TV distance mixing time for every $\varepsilon > 0$ by
\[
    \tau_{TV}(\varepsilon) := \min\{t: D_{(1)}(t) \le \varepsilon\}.
\]
\item 
For $p=2$, 
$\D{2}{t}$ captures the square root of $\chi^2$ divergence (also known as the variance) between $p^t$ and $\pi$. As such, we will also refer to $\D{2}{t}$ as $\chi$-divergence.
    \end{itemize}
\end{remark}

We remark that $\ell_p$ norms are monotone.
\begin{observation}
    $\D{1}{t} \le \D{2}{t} \le \D{\infty}{t}$.
\end{observation}

\paragraph{Uniform mixing times.}
Mixing times in terms of $\D{\infty}{t}$ are known as uniform mixing time~\cite{morris2005evolving}, which is defined as: 
\[
\tau_U(\eps) := \min \set{ t : \D{\infty}{t} \le \eps}.
\]
In the following, we show that uniform mixing follows from many standard mixing time results.

\subsection{Uniform mixing from $\ell_2$ mixing for reversible Markov chains}
Next we recollect a standard connection between $\ell_2$ and $\ell_\infty$ mixing, which appeared implicitly in~\cite[Corollary 2.1.5]{gine1997lectures} and explicitly in~\cite[Proposition 4.15]{levin2017markov}. 
For completeness we also provide an elementary proof.
\begin{proposition} For any reversible Markov chain,
    $\D{\infty}{2t} = \tuple{\D{2}{t}}^2 = \max_{x\in \Omega} q_{2t}(x,x) - 1$.
\end{proposition}
\begin{proof}
    Notice that $p^{r+t}(x,z) = \sum_{y\in \Omega} p^r(x,y) p^t(y,z)$,  and for stationary distribution $\pi$ we have
\begin{align*}
p^{r+t}(x,z) - \pi(z) =&
\sum_{y\in \Omega} p^r (x,y) \tuple{p^t (y,z) - \pi(z) }\\
=&\sum_{y\in \Omega} \tuple{p^r (x,y) - \pi(y) } \tuple{p^t (y,z) - \pi(z) },
\end{align*}
where the first equation follows from stochasticity of the transition probabilities, and the second equation follows from stationarity of $\pi$.

Then by reversibility of the Markov chain, we have that for every $x,z\in \Omega$,
\begin{align}
\frac{p^{r+t}(x,z)}{\pi(z)} - 1 = &\sum_{y\in \Omega} \pi(y) \tuple{\frac{p^r (x,y)}{\pi(y)} -1}  \tuple{\frac{p^t (y,z)}{\pi(z)} -1} \nonumber \\
\overset{\textrm{(reversibility)}}{=} &\sum_{y\in \Omega} \pi(y) \tuple{\frac{p^r (x,y)}{\pi(y)} -1} \tuple{\frac{p^t (z,y)}{\pi(y)} -1} \label{eq:mc-reversible-rearranged} \\
\overset{\textrm{(Cauchy-Schwarz)}}{\le} &\sqrt{\sum_{y\in \Omega} \pi(y) \tuple{\frac{p^r (x,y)}{\pi(y)} -1}^2}\cdot  
\sqrt{\sum_{y\in \Omega} \pi(y) \tuple{\frac{p^t (z,y)}{\pi(y)} -1}^2} \nonumber \\
\le& \D{2}{r} \D{2}{t}. \nonumber
\end{align}
By choosing $r=t$ we have $\D{\infty}{2t} \le \tuple{\D{2}{t}}^2$. By setting $x=z$ in~\eqref{eq:mc-reversible-rearranged} above we have
\begin{align*}
\frac{p^{2t}(x,x)}{\pi(x)} - 1 = \sum_{y\in \Omega} \pi(y) \tuple{\frac{p^t (x,y)}{\pi(y)} -1}^2.
\end{align*}
Combined, we have
\[
\D{\infty}{2t} \le \tuple{\D{2}{t}}^2 = \max_{x\in \Omega} \frac{p^{2t}(x,x)}{\pi(x)} - 1= \max_{x\in \Omega} q_{2t}(x,x) - 1.
\]
Recalling the definition of $\D{\infty}{2t}$, one must have $\D{\infty}{2t}\ge \max_{x\in \Omega} q_{2t}(x,x) - 1$, thus the inequality must also be equality, which concludes the proof.
\end{proof}
\begin{remark}
    We note that the only place that we use reversibility of the Markov chain is that, the time reversal of the Markov chain also mixes in $\ell_2$ equally fast. One could apply the above argument using $\ell_2$ mixing time of the time reversal process (if known, e.g., by a conductance based approach) to get rid of the reversibility requirement. 
\end{remark}
    For an irreducible Markov chain with transition matrix $P$, let $\gamma_*$ be the absolute spectral gap given by 
    \[
    \gamma_* = \min\set{ 1- \abs{\lambda}: \lambda\neq 1 \hbox{ is an eigenvalue for }P}.
    \]
\begin{proposition}
\label{prop:linfty-l2}
For a reversible and irreducible Markov chain with absolute spectral gap $\gamma_*$, and let $\pi_*=\min_{x \in \supp(\pi)} \pi(x)$ be the smallest probability in the stationary distribution,
then $\tau_U(\eps) \le \frac{2}{\gamma_*} \ln \tuple{\frac{1}{\eps \pi_*}}$.
\end{proposition}
\begin{proof}
By the previous proposition,
\[
\D{\infty}{2t} = \tuple{\D{2}{t}}^2 \le \frac{(1-\gamma_*)^t}{\pi_*},
\]
we get $\D{\infty}{2t} \le \eps$ by setting $t=\frac{1}{\gamma_*} \ln \tuple{\frac{1}{\eps \pi_*}}$.
\end{proof}
\subsection{Uniform mixing from $\ell_1$ mixing}
We note that the standard separation distance is indeed closely related to the $\ell_1$ distance (see, e.g., \cite[Lemma 6.17]{levin2017markov}). However, separation distance is asymmetric, and what we actually need in our perfect sampling scheme is the ``reverse'' direction of separation distance.

We give two arguments for how $\ell_1$ mixing implies uniform mixing. Firstly, we argue that by setting $\eps = \eps' \pi_*$, where $\pi_* = \min_{x \in \supp(\pi)} \pi(x)$,
then $\eps$-TV distance mixing implies a $\eps'$-uniform mixing.
Notably, reversibility of the Markov chain is not needed  here.

\begin{proposition}
\label{prop:linfty-l1}
   $ \D{\infty}{t} \le \frac{\D{1}{t}}{\pi_*}$.
\end{proposition}
\begin{proof}
For any $x,y \in \Omega$,
    \[
    \abs{\frac{p^t(x,y)}{\pi(y)} -1} \le \frac{\abs{p^t(x,y) - \pi(y)}}{\pi_*} \le \frac{\D{1}{t}}{\pi_*}.
    \]
    Taking the maximum yields $ \D{\infty}{t} \le \frac{\D{1}{t}}{\pi_*}$.
\end{proof}

Secondly, by further assuming reversibility of the Markov chain, we note that a polynomial $\ell_1$ mixing bound is a certificate to an inverse polynomial spectral gap for a lazy and reversible Markov chain. This allows us to use the previous subsection to conclude uniform mixing as well. Specifically, let $\tau_{TV}(1/4)$ be the first time such that $\D{1}{t}\le 1/4$, one can show the following.

\begin{proposition}[\text{\cite[Theorem 12.5]{levin2017markov}}] \label{prop:mix=var}
  Let $P$ be a reversible, irreducible, and aperiodic Markov chain.
  It holds that the absolute spectral gap $\gamma_* = \Omega\tuple{\tau_{TV}(1/4)}$.
\end{proposition}

%% file: mc-sampler.tex
\section{Perfect sampling from uniform mixing}
As an application of the generic reduction of \Cref{thm:gen-red}, we can derive a simple perfect sampling scheme from mixing time in terms of $\D{\infty}{t}$.

\begin{corollary}
\label{cor:mc-sampler}
Let $(X_t)$ be any Markov chain with state space $\Omega$ and uniform mixing time bound $\tau_U(\eps)$. Furthermore, assume that there exist:
\begin{itemize}
\item an algorithm $\markovstep$ that can sample $X_{t+1} \sim P(X_t, \cdot)$ given $X_t$, and,
\item an algorithm $\transitioncompalgo$ that outputs the transition matrix $P$.
\end{itemize}
Then, there is a perfect sampler for the stationary distribution with expected running time 
\[
O\left(\tau_U\left(\frac{1}{|\Omega|^4}\right) \cdot T(\markovstep) + \frac{1}{|\Omega|^4} \cdot T(\transitioncompalgo)\right).
\]

Furthermore, it terminates within $O\left(\tau_U\left(\frac{1}{|\Omega|^4}\right) \cdot T(\markovstep)\right)$ time except with probability $\frac{1}{|\Omega|^4}$. 

\end{corollary}

\begin{proof}
This follows from \Cref{thm:gen-red} by initiating $\targetdist, \eps, \approxsampler_\eps, \targetcomp_{\targetdist}, \samplercomp_{\approxsampler_\eps}$ as follows:
\begin{itemize}
\item Let $\targetdist$ be $\pi$, the stationary distribution.
\item Let $\eps$ be $1/|\Omega|^4$.
\item Let $\approxsampler_\eps$ be $X_{\tau_U\left(\eps\right)}$ when starting with $X_0 = x$, where each step of sampling is done via $\markovstep$. Note that $T(\approxsampler_\eps) \leq O\left(\tau_U\left(\eps\right) \cdot T(\markovstep)\right)$
\item Let $\targetcompall_{\targetdist}$ be the following algorithm: run $\transitioncompalgo$ to get $P$, and then use Gaussian elimination\footnote{If bit complexity is a concern, Edmond's version of Gaussian elimination~\cite{Edmonds1967SystemsOD} guarantees efficient bit complexity.} to solve for $\pi$ using the equations $\pi P = \pi$ and $\pi \mathbf{1}^\top = 1$. Note that $T(\targetcompall_{\targetdist}) \leq O(T(\transitioncompalgo) + |\Omega|^3)$. 
\item Let $\samplercompall_{\approxsampler_\eps}$ be the following algorithm: run $\transitioncompalgo$ to get $P$, and, for all $x \in \Omega$, output $\mathbf{1}_x^\top P^t$ for $t = \tau_U\left(\eps\right)$. We have $T(\samplercompall_{\approxsampler_\eps}) \leq O(T(\transitioncompalgo)  + \tau_U\left(\eps\right) \cdot |\Omega|^4)$.
\end{itemize}
As a result, applying \Cref{thm:gen-red}, we have a perfect sampler for $\pi$ with expected running time $$O\left(T(\approxsampler_\eps) + \eps \cdot T(\samplercompall_{\approxsampler_\eps}) + \eps \cdot T(\targetcompall_{\targetdist}) + \eps \cdot |\Omega|\right) \leq O\left(\tau_U\left(\eps\right) \cdot T(\markovstep) + \frac{1}{|\Omega|^4} \cdot T(\transitioncompalgo)\right).$$

For the high probability statement, notice that the perfect sampler has to call  $\samplercompall_{\approxsampler_\eps}$ and $\targetcompall_{\targetdist}$ only with probability $\eps =1/|\Omega|^4$.
\end{proof}
\begin{remark}
    We note that the choice of $\eps=1/|\Omega|^4$ is arbitrary and suffices for all our applications. If $T(\transitioncompalgo) \gg |\Omega|^4$ one could also choose $\eps = 1/T(\transitioncompalgo)$.
\end{remark}

Now we are ready to complete the proof of our main theorem.

\begin{proof}[Proof of~\Cref{thm:main}]
    Given an $\ell_1$ mixing time $T$, by the submultiplicativity of $\ell_1$ mixing time and \Cref{prop:linfty-l1}, we have $\tau_U(\eps) \le T \ln \frac{1}{\eps \pi_*}$. By \Cref{cor:mc-sampler}, we have a perfect sampler with expected running time $O\tuple{T \ln \frac{\abs{\Omega}}{\pi_*}}$.

    Given an absolute spectral gap $\gamma_*$, by \Cref{prop:linfty-l2} we have $\tau_U(\eps) \le \frac{1}{\gamma_*} \ln \frac{1}{\eps \pi_*}$. Thus by \Cref{cor:mc-sampler}, we have a perfect sampler with expected running time $O\tuple{\frac{1}{\gamma_*} \ln \frac{\abs{\Omega}}{\pi_*}}$.
\end{proof}
